\documentclass[11pt]{article}

\usepackage{a4wide,amsmath,amssymb,amsthm}
\usepackage[utf8]{inputenc}
\usepackage{enumerate}
\usepackage{latexsym,epsfig,graphicx}
\usepackage[x11names]{xcolor}
\usepackage[]{booktabs}
\usepackage{multirow,multicol}
\usepackage{url}
\usepackage[authoryear]{natbib}
\usepackage{complexity}
\usepackage{setspace}
\usepackage{indentfirst}
\usepackage{float}
\usepackage[flushleft]{threeparttable}
\usepackage{placeins}
\usepackage[linesnumbered,ruled,vlined,english]{algorithm2e}
\usepackage[affil-it]{authblk}
\usepackage{csquotes}
\usepackage{nicefrac}

\newtheorem{theorem}{Theorem}
\newtheorem{proposition}{Proposition}

\colorlet{linkequation}{blue}
\usepackage[colorlinks]{hyperref}

\newcommand*{\SavedEqref}{}
\let\SavedEqref\eqref
\renewcommand*{\eqref}[1]{
\begingroup
\hypersetup{
    linkcolor=linkequation,
    linkbordercolor=linkequation,
}
\SavedEqref{#1}%
\endgroup
}

\hypersetup{
    citecolor=SpringGreen4,
    citebordercolor=SpringGreen4,
}

\title{The matching relaxation for a class of generalized set partitioning problems}

\author[1]{Phillippe Samer}
\author[1]{Evellyn Cavalcante}
\author[1]{Sebasti\'{a}n Urrutia}
\author[2]{Johan Oppen}
\affil[1]{\footnotesize Universidade Federal de Minas Gerais\\Av. Ant\^onio Carlos 6627, 31270-010, Belo Horizonte, MG, Brazil \newline {\tt phillippes@gmail.com, evellyn.cavalcante@gmail.com,  surrutia@dcc.ufmg.br}}
\affil[2]{\footnotesize Molde University College\\P.O.box 2110, N-6402, Molde, Norway \newline {\tt johan.oppen@himolde.no} }

\date{May 11, 2018}

\begin{document}

\maketitle

\begin{abstract}
This paper introduces a discrete relaxation for the class of combinatorial optimization problems which can be described by a set partitioning formulation under packing constraints.
We present two combinatorial relaxations based on computing maximum weighted matchings in suitable graphs. Besides providing dual bounds, the relaxations are also used on a variable reduction technique and a matheuristic.
We show how that general method can be tailored to sample applications, and also perform a successful computational evaluation with benchmark instances of a problem in maritime logistics.


{\bf Remark:} a preliminary, four-page abstract of this work was presented at the 14th Co\-log\-ne-Twen\-te Workshop on Graphs and Combinatorial Optimization, 2016.

{\bf Acknowledgements:} this work was carried out within the Port-Ship Coordinated Planning project, supported by the Norwegian Research Council under project number 227084/O70. We also thank the three referees whose comments improved our presentation.

\end{abstract}

\section{Introduction}
\label{sec:intro}

Consider the following problem, which the reader might recognize from a range of application domains.
A set of $n$ tasks is to be performed, and the problem solver has to decide for an execution mode to them.
To each individual task corresponds a particular set of acceptable execution modes to perform it. Moreover, each execution mode comprises a well-defined cost and resource usage (possibly many, \emph{e.g.} time, space, tools, workers). Since the finite resources are to be shared among tasks, the problem solver also counts on an oracle capable of determining whether a selection of execution modes to different tasks is \emph{compatible}. Any solution prescribing the assignment of incompatible execution modes is thus rendered infeasible.

We concern the class of problems which, at its core, can be cast as follows. Given sets of acceptable assignments to each individual task, and compatibility information among any selection of assignments, find a compatible setting for the $n$ tasks of minimum total cost. 
To clearly outline our contributions and the scope of our investigation, we first give a precise formulation of the class of problems we concern, which we refer to as the base generalized set partitioning problem (GSPP).

\subsection{The base GSPP}

Let $T$ be the set of tasks, and $R = R_1 \times R_2 \times \ldots \times R_k$ be the set of tuples identifying resource usage, \emph{i.e.} combinations of an option for resource $R_1$, an option for resource $R_2$, and so forth.
Also define $\Omega_i$ as the set of feasible assignments to task $i$: each element in $\Omega_i$ assigns a subset of resources in $R$ (characterizing an execution mode) to complete an individual task.
We denote $\Omega = \cup_{i=1}^n \Omega_i$.
Binary decision variables $\mathbf{y} \in \mathbb{B}^{|\Omega|}$ thus indicate which individual assignments are used in the solution.
Let $c_j \in \mathbb{Q}$ denote the cost of an assignment $y_j$ to an individual task, consisting of a computable function of the total resource usage on that individual assignment alone. 
Finally, the coefficient matrices are as follows.
$A \in \mathbb{B}^{|T|\times|\Omega|}$ associates each column with a single task: $a_{ij}$ is equal to one if $y_j$ refers to an assignment for task $i$; otherwise, it is equal to zero.
$B \in \mathbb{B}^{|R|\times|\Omega|}$ represents resource usage tuples: $b_{rj}$ is one iff the given combination $r \in R$ of individual resources is used in the assignment $y_j$.
Then, we build on the following integer programming (IP) formulation:

\begin{equation}
z = \min \left\lbrace \sum_{j \in \Omega}{c_j y_j} : \textbf{y} \in \mathcal{P}_{\text{gspp}} \cap \mathbb{B}^{|\Omega|}  \right\rbrace,
\label{gspp:obj}
\end{equation}
where $\mathcal{P}_{\text{gspp}}$ denotes the polyhedral region defined by:
\begin{alignat}{2}
& \sum_{j \in \Omega}{a_{ij} y_j = 1}  &  \hspace{2cm} \forall i \in T    \label{gspp:partition} \\
& \sum_{j \in \Omega}{b_{rj} y_j \leq 1}  &  \forall r \in R   \label{gspp:pack} \\
& y_{j} \leq 1  &  \forall j \in \Omega \label{gspp:yUB} \\
& y_{j} \geq 0  &  \forall j \in \Omega \label{gspp:yLB}
\end{alignat}
Set partitioning constraints $\eqref{gspp:partition}$ ensure that all tasks are served by exactly one assignment, while set packing in $\eqref{gspp:pack}$ forbids overlapping of assignments in each resource combination slot.
We note that, if incompatibilities among tasks reduced to pairwise relations, the latter class of inequalities could be replaced by any set packing relaxation, such as \emph{edge inequalities}: $x_u + x_v \leq 1$, for each edge $(u,v)$ in a conflict graph \citep{atamturk2000conflict}.
Constraints $\eqref{gspp:partition}-\eqref{gspp:yLB}$ correspond to the linear relaxation of the binary programming formulation.

This formulation has connections to different disciplines in combinatorial optimization.
It is similar to some variations of the assignment problem \citep{pentico2007assignment}.
It can be seen as the scheduling of jobs on parallel machines minimizing total processing time, as we illustrate in Section \ref{sec:app:scheduling} with a problem studied by \cite{LallaRuizScheduling2016}.
And the formulation is also an instance of the mixed set covering, packing and partitioning problem, investigated by \cite{kuo2016}, following a longer tradition of studying perfect and ideal 0--1 matrices.
The authors perform the first polyhedral investigation of the mixed problem, and argue on its relevance and number of applications, notwithstanding the fact that it has drawn little attention in the literature so far.
For the interested reader, we indicate in Section \ref{sec:relatedwork} more related problems and situate the above GSPP structure in the set partitioning literature.

\subsection{Our contributions}
The number of variables in the previous formulation can be huge, although polynomial in the number of execution modes of tasks.
This is in consonance with the compromise between: (i) the complexity in the representation of each execution mode to a task, \emph{i.e.} the level of discretization, and (ii) the computing time to solve the resulting problem.

The main idea of this paper is to show that one can use interesting, combinatorial constructions over the variables to find lower bounds to $\eqref{gspp:obj}$.
Moreover, since the bounds are purely combinatorial, they may be computed more quickly than bounds based on linear programming (LP) relaxation, though possibly at the expense of being weaker.
As noted by
\cite{ryan1992}, \emph{``The solution of the LP relaxation has always proved to be a computational bottleneck in solving SPPs''}.
Even if the combinatorial bounds prove to be weaker on a given application, the construction itself might be interesting as a building block, \emph{e.g.} in algorithms which depend on the computation of dual bounds, as we illustrate with a variable probing method and a matheuristic in Section~\ref{sec:algorithms}.

To summarize, the contributions of this paper are:
\begin{itemize}

\item
We draw attention to the GSPP structure $\eqref{gspp:obj}-\eqref{gspp:yLB}$ as being interesting in its own right, and list a range of applications of it in the recent literature.
Moreover, the key ideas we present can be extended to fit application-specific details, as we illustrate in Section~\ref{sec:applications}. 

\item In Section \ref{sec:relaxations}, we introduce two lower bounds to problem $\eqref{gspp:obj}$. These lower bounds come from weighted matchings computed on suitable graphs. We also prove that one of the bounds is stronger than the other.

\item From an algorithmic standpoint, the combinatorial lower bounds are embedded in a model-based heuristic framework (or \emph{math}euristic), presented in Section \ref{sec:algorithms}. 
We argue that the efficient bounds introduced are key to this algorithm.

\item A preliminary computational evaluation of that matheuristic, using benchmark instances from a logistics problem, indicate that it is able to find near-optimal solutions in reduced execution time (Section \ref{sec:results}).

\end{itemize}

\vspace{-0.5cm}

\section{On the literature of generalized set partitioning problems}
\label{sec:relatedwork}
In this section, we give pointers to the literature of the standard set partitioning problem first, followed by a brief account of the different generalizations related to the structure we investigate.

\subsection*{The standard set partitioning problem}
The first paragraph in Section 1 of a classic review by \cite{balas1976set} reads: 
\begin{displayquote}
``Among all special structures in (pure) integer programming, there are three which have the most wide-spread applications: set partitioning, set covering and the traveling salesman (or minimum length Hamiltonian cycle) problem; and if we were to rank the three, set partitioning would be a likely candidate for number one.''
\end{displayquote}

We remark that, over forty years later, the set partitioning problem is still a central structure in integer programming, as illustrated by the several interesting applications we mention throughout this paper.
That early work of Balas and Padberg is a thorough introduction to the concepts and problems related to set partitioning, as well as the main polyhedral results and exact algorithms up to that date. While surveying the literature on this fundamental problem is beyond the scope of our paper, we still mention particular works that are more directly related to ours or particularly inspiring.

The awarded thesis by \cite{borndorfer1998Thesis} provides an extended review of polyhedral and algorithmic aspects related to the underlying set packing and covering relaxations. It also documents the components and implementation issues of a branch and cut algorithm for set partitioning, extending in numerous directions the approach of \cite{hoffman1993airline} to an application in airline crew scheduling. Finally, it describes a real world application of vehicle routing in the context of public transportation for the physically disabled, building on a decomposition of the problem in two steps, both of which require optimizing over a set partitioning structure.
In the next two paragraphs, we highlight the remarkable effort of that author in extending the results on polytopes with more substantial body of knowledge to larger classes of problems.

\subsection*{Polyhedra and combinatorics of set partitioning}
\cite{BorndorferWeismantel2001} present an affine transformation technique (\emph{aggregation}) to leverage cutting planes from combinatorial relaxations of an IP formulation.
Their polyhedral investigation stems from a generalization of projection to introduce an algorithmic approach. A so-called \emph{aggregation scheme} allows one to transfer (\emph{expand}) known classes of valid inequalities in a suitable projection space (\emph{e.g.} the polytope of another combinatorial structure) to the original polyhedron representing an IP.
Besides constructing interesting aggregation schemes leading to set packing and knapsack relaxations, the authors describe conditions under which the separation problem for the new classes of valid inequalities can be solved in polynomial time.
Several of those results are further investigated and contrasted within the framework of disjunctive cuts by \cite{Letchford2001}.

\cite{2004BorndorferChapter} also describes the combinatorial packing problem (CPP), investigating a structure closely related to the class of set partitioning problems that we study. The problem consists of solving a number of individual combinatorial optimization problems on the same ground set, such that no element is contained in the solution of more than one problem. A number of interesting problems can be subsumed as CPP examples, including minimum cost flows, steiner trees packing, and the generalized assignment problem.

The long tradition in studying the polyhedra associated with 0--1 matrices encompasses the GSPP structure that we consider.
In particular, \cite{SeboIPCO1998} extended the theory of perfect matrices in set packing formulations, and ideal matrices in set covering, to characterize non-integral polyhedra in the mixed packing and covering problem (assuming the coefficient matrices satisfy those former conditions). 

A central structure in the work of Seb{\H{o}} is the odd-hole graph.
Almost twenty years later, \cite{kuo2016} start from that same graph structure to investigate the mixed set covering, packing and partitioning problem.
They derive the \emph{mixed odd-hole inequality}, and show that its inclusion completely characterizes the polytope corresponding to the mixed problem when the coefficient matrix induces precisely that graph.
As we remarked earlier, the base GSPP formulation that we explore in this work is an instance of that mixed problem.

\subsection*{Formulations and algorithms for generalized set partitioning problems}

There is a solid research track on Lagrangean bounds and algorithms for generalizations of set partitioning.
\cite{nemhauser1979optimal} introduced a reformulation of the classic set partitioning problem as a weighted matching with simple side constraints. They describe a Lagrangean relaxation algorithm, exploiting the particular matching structure to improve a standard nonlinear optimization method to solve dual subproblems.
\cite{ali1989network} also present a Lagrangean approach to solve relaxations within a branch and bound framework. The original problem structure is decomposed into a network-type submatrix (solved efficiently with dual reoptimization), while the remaining constraints are dualized in the objective function, resulting in an integral relaxation.
\cite{shor1996using} further explore Lagrangean bounds to formulate relaxations of set packing/partitioning problems. They even consider an approximation procedure solving a reduced IP model, which is analogous to the \emph{math}euristic we present in Section \ref{sec:matheuristic}, using Lagrangean duals instead of the combinatorial bounds we introduce.
\cite{el1995graph} review and propose a series of combinatorial relaxations for the classic set partitioning formulation. They also suggest that their results can be used in a Lagrangean relaxation approach. 

\cite{linderoth2001} present a parallel algorithm for distributed memory architectures, combining different primal and dual LP-based heuristics with a collection of preprocessing, variable fixing and cut generation methods.

Finally, we mention a few more problems that are closely related to the GSPP structure that we investigate here.
In the work of \cite{Campello1987}, the set partitioning problem is generalized by a single upper bound on positive linear combinations of the variables, instead of the set packing constraints of the GSPP we study. They also consider a different cardinality constraint, and propose a Lagrangean relaxation scheme to determine lower bounds.
\cite{ball1990matching} describe the Lagrangean relaxation of weighted matchings with generalized upper bounds, which can be cast as a particular instance of the set partitioning problem. They present a thorough experimental analysis, considering different approaches to optimize the Lagrangean dual problem, an improvement heuristic, and 
a convergent enumeration procedure to guarantee the optimality of the overall method.
As for \cite{Fisher1990}, the set partitioning structure is generalized by set covering constraints. The resulting problem is also interesting, providing a model for previous applications in the literature, and the authors describe dual heuristics to it.
\cite{chan1992} and \cite{boschetti2008} also present dual heuristics to the problem.
Another interesting Lagrangean approach worth mentioning is that of \cite{CavalcanteCidAbilio2008} to set partitioning. Their sophisticated relax-and-cut algorithm improves on the quality of previous known lower bounds, while also being competitive regarding time efficiency.

\section{Matching relaxations of the GSPP}
\label{sec:relaxations}

In this section, we present the main contributions of the paper.

For the sake of clearness, we highlight from the problem definition in the Introduction three conditions for our key ideas to work. First, we assume that each individual assignment $y_j$ has a cost $c_j$ which is computable independently of the assignments to other tasks.
We also assume that it is possible to determine whether a selection of individual assignments is compatible.
Finally, the GSPP structure we study is limited to formulations with a polynomial number of variables, such that the set $\Omega$ can be enumerated before solving the resulting IP. While this immediately rules out a series of applications (typically solved by column generation algorithms), we hope that our discussions and numerical results in the remainder of the paper  could settle the relevance of the class of problems we concern.

We start by introducing two relaxations for the GSPP formulation, which yield dual bounds to the optimal value of the objective $z$ in $\eqref{gspp:obj}$.
Note that the most na\"ive approach would be to discard all the packing constraints in $\eqref{gspp:pack}$ and simply choose the cheapest individual assignment to each task, which would provide a most trivial lower bound. In the following, we aim to discard less of those constraints.
We construct two simple, undirected graphs, representing a subset of the enumerated assignments.  
Throughout the text, we use the linear map $c: \Omega \rightarrow \mathbb{Q}$ from the space of assignments to their costs, such that $c(y_j) = c_j$.


We define the graph $G_1(T,E_1)$, with a vertex for each task. The set $E_1$ includes an edge $(i,j)$ if the individual assignments of best cost for tasks $i$ and $j$ are not compatible with each other.
Let $c^{\prime}_j$ denote the minimum cost assignment for task $j$; that is, $c^{\prime}_j = \min \{ c(y_j): y_j \in \Omega_j \}$. Analogously, let $c^{\prime\prime}_j$ be the second minimum cost assignment for $j$.
The cost $c_1(i,j)$ of an edge in $G_1$ is defined by the least difference among such costs, for the corresponding tasks $i$ and $j$. That is: $c_1(i,j) \triangleq \min \{(c^{\prime\prime}_i - c^{\prime}_i), (c^{\prime\prime}_j - c^{\prime}_j) \}$.
Then, the following bound on the cost of any feasible solution holds.

\begin{theorem}
\label{theorem:lb1}
Let $M \subseteq E_1$ denote a maximum weighted matching in $G_1$, and $w(M) = \sum_{e \in M} c_1(e)$ be its weight. Then $LB_1 \triangleq w(M) + \sum_{j \in T} c^{\prime}_j $ is a lower bound to the optimal value $z$ in $\eqref{gspp:obj}$.
\end{theorem}
\begin{proof}
The selection of the best individual assignments for each task corresponds to relaxing all the constraints in $\eqref{gspp:pack}$. Therefore, this is a trivial lower bound to the cost of any feasible solution, and amounts to $\sum_{j \in T} c^{\prime}_j$.

Starting with the trivial selection of best individual assignments, 
and remembering that $(i,j) \in E_1$ iff such best assignments for $i$ and $j$ are not compatible,
the weight of edge $(i,j)$ corresponds to the minimum cost increase due to exchanging one such assignment for the second best. Clearly, this new pair of assignments for tasks $i$ and $j$ can still be infeasible, but the sum of their costs is a lower bound to the cost of any compatible assignment for these tasks.

Note that we cannot imply that the accumulated costs of edges incident to a given vertex $i$ are necessary, because the graph does not provide information about which of the extremes of an edge $(i,j)$ assumes the second best assignment,
\emph{i.e.} whether $c_1(i,j) = (c^{\prime\prime}_i - c^{\prime}_i)$ or $c_1(i,j) = (c^{\prime\prime}_j - c^{\prime}_j)$.
It is even possible that, following such an exchange for a second-best assignment, other edges would not even exist in $G_1$.
However, one can consider any matching in $G_1$, corresponding to disjoint pairs of tasks whose best assignments are not compatible. Therefore, the weight of any matching is a required cost increase over $\sum_{j \in T} c^{\prime}_j$, implied by the pairwise incompatibility of the corresponding individual assignments. In particular, a maximum weighted matching corresponds to the strongest such bound in $G_1$.
\end{proof}


Our second dual bound strengthens the information on the cost of compatible assignments between pairs of tasks. 
Let $G_2(T,E_2)$ denote a complete graph, with a vertex for each task. 
Although the following result holds for any number of tasks, it would be unnecessarily weaker for odd $|T|$ because the new bound amounts solely to the weight of a matching in this new graph, and some vertex would not be covered. To circumvent this, in the case that $|T|$ is odd, we simply add to $G_2$ an artificial vertex $s$, with edges to every other vertex $i$, with costs $c_2(s,i) = \min \{ c(y_i) : y_i \in \Omega_i \}$.
The remainder of the paper thus assumes that $G_2$ and the corresponding bound are defined over an even number of vertices.

Define the cost $c_2(i,j)$ of an edge in $E_2$ as the cheapest compatible assignments for tasks $i$ and $j$, that is:
$c_2(i,j) \triangleq \min \{ c(y_i) + c(y_j): y_i \in \Omega_i, y_j \in \Omega_j, y_i \text{ and } y_j \text{ are compatible} \}$.  Then, we have the following result.

\begin{theorem}
\label{theorem:lb2}
Let $M \subseteq E_2$ be a maximum weighted matching in $G_2$. Then,
$LB_2 \triangleq \sum_{e \in M} c_2(e)$ is a lower bound to the optimal value $z$ in $\eqref{gspp:obj}$.
\end{theorem}
\begin{proof}
The weight of a single edge $(i,j) \in E_2$ is the sum of the minimum cost assignments for tasks $i$ and $j$, complying with the compatibility constraints \emph{among them}. That is: these two assignments alone are compatible. A selection of edges not sharing a vertex (\emph{i.e.} a matching) thus corresponds to pairing up tasks and determining their best compatible assignments, which is required in any solution satisfying packing constraints $\eqref{gspp:pack}$.
Therefore, the weight of any matching in $G_2$ is a lower bound to the cost of a feasible solution, since this clearly relaxes constraints regarding the compatibility of unpaired tasks. A maximum weighted matching thus provides the strongest such bound in $G_2$.
\end{proof}



We conclude with a result on the relative strength of the bounds obtained in the two relaxations.
Note that, in the simple case where all the best individual assignments are pairwise compatible with each other, we verify:
(i) the graph $G_1$ has no edges, and the bound $LB_1$ corresponds to trivial bound $\sum_{j \in T} c^{\prime}_j$ ;
(ii) any perfect matching $M$ in the graph $G_2$ has maximum weight, amounting to the sum of costs of the best individual assignments.
Therefore, the bounds are equal: $LB_2 = \sum_{(i,j) \in M} c_2(i,j) = \sum_{(i,j) \in M} (c^{\prime}_i + c^{\prime}_j) = \sum_{u \in T} c^{\prime}_u = LB_1$.
We show below that the second bound is actually stronger than the first.

\begin{theorem}
\label{theorem:strength}
The lower bound attained from graph $G_2(T,E_2)$ is stronger than that from graph $G_1(T,E_1)$; \emph{i.e.} for any given problem instance, $LB_2 \geq LB_1$ holds, and $LB_2 > LB_1$ for at least one case.
\end{theorem}
\begin{proof}
First, we remark that specific cases where $LB_2 > LB_1$ are intuitive. It suffices to have a pair of vertices for which there are no compatible assignments employing the cheapest execution mode for one of them.
In the following, we suppose there is an instance where $LB_1 > LB_2$. We build a matching in $G_2$ with cost at least $LB_1$, showing that the hypothesis is absurd.

For any edge $(i,j) \in E_1$, we can compare the cost functions in $G_1$ and $G_2$; recall that the latter graph is complete.
By definition, $c_1(i,j)$ corresponds to the minimum cost increase implied by the incompatibility of the best individual assignments for $i$ and $j$, while $c_2(i,j)$ corresponds to the actual sum of the costs of the best compatible assignments.
It follows that:
\begin{equation}
c_2(i,j) \geq c^{\prime}_i + c^{\prime}_j + c_1(i,j) \label{eq:prova3:custos}
\end{equation}

Let $M_1 \subseteq E_1$ be a maximum weighted matching in $G_1$, which thus yields the lower bound $LB_1$ from that graph.
We define the analogous set of edges in $G_2$ as $M_2 = \{ (i,j) \in E_2 : \text{there exists the edge } (i,j) \in M_1 \}$.
The set $M_2$ is a matching in $G_2$, by construction.

We distinguish two cases. If $M_1$ is perfect, then $M_2$ is perfect as well since both graphs have the same vertex set. We can infer about their weights:
\begin{equation}
w(M_2) = \sum_{(i,j) \in M_2} c_2(i,j) \geq \sum_{(i,j) \in M_2} {(c^{\prime}_i + c^{\prime}_j + c_1(i,j))} = \sum_{(i,j) \in M_1} {c_1(i,j)} + \sum_{u \in T} c^{\prime}_u = LB_1, \label{eq:prova3:perfeito}
\end{equation}
where the first inequality holds by $\eqref{eq:prova3:custos}$, and the second equality is true because the matchings are perfect.

If $M_1$ is not perfect, there are pairs of vertices $(x,y)$ not covered by $M_1$.
By hypothesis, $M_1$ has maximum weight; hence $(x,y) \not\in E_1$, \emph{i.e.} the individual assignments of least cost for $x$ and $y$ are compatible.
Therefore, the edge in $G_2$ corresponding to each such pair $(x,y)$ has cost $c_2(x,y) = c^{\prime}_x + c^{\prime}_y$.
We can extend $M_2$ to a perfect matching $M_2^{\prime}$ in $G_2$ by arbitrarily connecting pairs of vertices not yet covered by $M_2$.
Let $C$ denote the set of edges selected this way, such that $M_2^{\prime} \triangleq M_2 \cup C$, and $\sum_{(x,y) \in C} c_2(x,y) = \sum_{(x,y) \in C} c^{\prime}_x + c^{\prime}_y$.
Then, analogously to the previous case, we have:
\begin{align}
w(M_2^{\prime}) &= \sum_{(i,j) \in M_2} c_2(i,j) + \sum_{(x,y) \in C} c_2(x,y) \nonumber \\
 & \geq \sum_{(i,j) \in M_2} {(c^{\prime}_i + c^{\prime}_j + c_1(i,j))}  + \sum_{(x,y) \in C} c_2(x,y) \nonumber \\
 &= \sum_{(i,j) \in M_2} {c_1(i,j)} + \sum_{u \in T} c^{\prime}_u \nonumber \\
 &= LB_1, \label{eq:prova3:naoperfeito}
\end{align}
where the last equalities hold because $M_2^{\prime}$ covers all vertices.

Therefore, the matchings in $G_2$ built in both cases $\eqref{eq:prova3:perfeito}$ and $\eqref{eq:prova3:naoperfeito}$ have weight at least $LB_1$, providing a lower bound on $LB_2$, which is defined as the maximum weight of a matching in $G_2$.
Since we start with a general input instance, the hypothesis that $LB_1 > LB_2$ could hold is absurd, and we always verify that $LB_2 \geq LB_1$.
\end{proof}

\section{Embedding the relaxation in a matheuristic algorithm}
\label{sec:algorithms}

This section extends our key idea, the matching relaxation of the GSPP structure, into algorithmic results.
First, we derive in Section \ref{sec:probing} a preprocessing method to probe and discard assignments that imply a suboptimal solution, as it is done in the work of \cite{iris2015} in the context of a port logistics problem.
Next, we present in Section \ref{sec:matheuristic} a matheuristic algorithm to find approximate solutions to the problem in reduced computational time.

\subsection{Preprocessing method for variable reduction}
\label{sec:probing}

The results in the previous section yield lower bounds on the optimal value $z$ of problem $\eqref{gspp:obj}$, but can also be extended to a preprocessing method. The goal is to fix at null value (or, equivalently, remove) a number of decision variables in the resulting model after the enumeration of feasible assignments, while preserving any optimal solution.

It is worth remarking that, since this technique is applied prior to the model optimization, such a proposal can be integrated with any approach based on enumerating the variables of the GSPP formulation and solving the resulting model with an integer linear programming algorithm.
This strategy has already been adopted by \cite{iris2015}, using lower bounds implied by probing the selection of a single assignment or a pair of assignments for two different tasks.

The next result assumes that an upper bound to $z$ is available. 
First, we temporarily assume that a given assignment $y_k \in \Omega_k$ is fixed in the solution.
We define the complete graph $G_{2,k}(T\backslash\{k\}, E_{2,k})$.
We proceed exactly as in the construction of $G_2$ towards Theorem $\eqref{theorem:lb2}$ to ensure that $G_{2,k}$ has an even number of vertices.
The corresponding edge costs $c_{2,k}$ regard the best compatible assignments for two given tasks, which are also compatible with $y_k$. That is:
\begin{align}
c_{2,k}(i,j) = \min \{ & c(y_i) + c(y_j): y_i \in \Omega_i, y_j \in \Omega_j, \nonumber \\
                       & y_i \text{ and } y_j \text{ are compatible with each other and with } y_k \} \nonumber
\end{align}
Finally, we evaluate the increase on the lower bound $LB_2$ implied by fixing the assignment $y_k$,
after computing a maximum weighted matching as presented in Theorem $\eqref{theorem:lb2}$: if the new lower bound exceeds a known upper bound, we conclude that this assignment cannot be part of an optimal solution.
The result is summarized as follows.

\begin{proposition}
\label{prop:probing}
Let $LB_{2,k}$ denote the lower bound from Theorem \ref{theorem:lb2} determined over $G_{2,k}$.
Given any upper bound $UB$ to $z$ in $\eqref{gspp:obj}$, if $c(y_k) + LB_{2,k} > UB$, then there is no optimal solution which includes the assignment $y_k \in \Omega_k$, and the corresponding variable can be removed from the model.
\end{proposition}

Therefore, we have an iterative algorithm for removing unnecessary variables in the model, while preserving all optimal solutions of the problem. For each feasible assignment in the GSPP formulation, one need only evaluate the new lower bound as depicted above.

Note that an analogous method could be derived from Theorem \ref{theorem:lb1}. Nevertheless, it follows immediately from Theorem \ref{theorem:strength} that it cannot be stronger, \emph{i.e.} it cannot remove a variable which the result in Proposition \ref{prop:probing} does not.

\subsection{Combinatorial ranking matheuristic}
\label{sec:matheuristic}

We introduce next an algorithm belonging to the class of matheuristics, or model-based heuristics, which integrate heuristics and ma\-the\-ma\-ti\-cal pro\-gram\-ming methods \citep{matheuristics2009book, itor2015issue}.
Specifically, we employ the combinatorial relaxation bound 
$LB_{2,k}$ computed in Proposition \ref{prop:probing}
to obtain a reduced model, which is optimized next with an integer linear programming solver.
We remark that previous strategies in the matheuristics literature include solving a reduced model, \emph{e.g.} after the heuristic removal of variables \citep{FanjulPeyro2011, stefanello2015}. 
Even in the context of a subproblem of the logistics application that we use to illustrate our ideas (see Section \ref{sec:app:bacap}), \cite{Mauri2008} present an evolutionary approach to generate columns using dual values in the LP relaxation as a fitness measure,
and \cite{LallaRuizBAP2016} employ an exact solver to partially optimize components of a previous solution using the \textsc{popmusic} metaheuristic.
Also, the early work of \cite{shor1996using} already describes the optimization of a reduced model for the set packing/partitioning problems using Lagrangean bounds to rank variables.

First, note that every solution to the GSPP formulation consists of only $|T| \ll |\Omega| $ assignments.
One could wonder if there would be a fast method for classifying variables, such that high quality solutions could be consistently achieved using only a fraction of the best ranked variables.
In this context, we discard the optimality certificate, and seek a high quality solution in reduced computation time.

The core of the method we propose is depicted in Algorithm \ref{alg:matheuristic}, which ranks and selects a subset $\Omega_F$ of variables from the GSPP formulation. The selection builds on the combinatorial bound $LB_{2,k}$ from Proposition \ref{prop:probing}, denoted by $\Delta(k)$ in the algorithm, and computed on the loop starting at line \ref{code:delta}.
The set of selected variables $\Omega_F$ corresponds to a subset of the polyhedron $\mathcal{P}_{\text{gspp}}$, and optimizing over it provides an upper bound $\bar{z} \geq z$ to the original problem $\eqref{gspp:obj}$. As we indicate in preliminary computational results (Section \ref{sec:results}), this bound can match the optimal value for benchmark instances of an application in port logistics even when using a relatively small fraction of variables.

The algorithm parameters are as follows.
\begin{description}
    \item[$\sigma$] : the least fraction of the best ranked variables to include in the final model;
    \item[$\mu$] : the least number of variables corresponding to each task, which the algorithm should ensure (when available) in the final model.
\end{description}

The latter parameter $\mu$ ensures that each task has a number of assignment options (as selected in the loop starting at line  \ref{code:mu}), while $\sigma$ controls the selection of variables among those implying the best dual bounds (loop starting at line \ref{code:sigma}).
It would be natural to consider algorithm variations, \emph{e.g.} selecting an exact number of variables, or performing a statistical study of the parameters. Both tasks could be approached in future work.\\


\begin{algorithm}[ht]
\small
\setstretch{1.35}
\DontPrintSemicolon
\SetKwInOut{Input}{Input}
\SetKwInOut{Output}{Output}
\SetKwInOut{Param}{Parameters}

\Input{initial set of variables $\Omega$, number of tasks $|T|$}
\Output{set $\Omega_F$ of variables selected for the final model}
\Param{least fraction $\sigma$ of the best ranked variables, least number of variables per task $\mu$}
\BlankLine

$ \Omega_F \leftarrow \varnothing$ \;
$ \Delta(k) \leftarrow 0$ for each $y_k \in \Omega$ \;

\ForEach{variable $y_k \in \Omega$}{   \label{code:delta}

    Let $G_{2,k}$ be the graph of compatible assignments, defined in Proposition \ref{prop:probing} \;
    Let $M$ be a maximum weighted matching in $G_{2,k}$ \;
    \tcp{lower bound implied by using this variable; see Prop.\ref{prop:probing} }
    $\Delta(k) \leftarrow c(y_k) + \sum_{e \in M} c_2(e)$ \;
}

Let $L$ denote the list of variables in $\Omega$, ordered by increasing values of  $\Delta$ \;

\While{$|\Omega_F|/|\Omega| < \sigma $}{   \label{code:sigma}

    Let $d$ be the least $\Delta$ value of a variable in $L$ \;
    \ForEach{variable $y_k \in L$ with $\Delta(k) = d$}{
        $\Omega_F \leftarrow \Omega_F \cup \{ y_k \}$ \;
        $L \leftarrow L \backslash \{ y_k \}$ \;
    }
}

\ForEach{task $i = 1, \ldots, |T|$}{   \label{code:mu}
    \While{$|\{ y \in \Omega_F : y \text{ is an assignment for } i \}| < \mu $ $\mathbf{and}$ $L$ contains a variable referring to $i$}{
        Let $y_i \in L$ be a variable referring to $i$, of least $\Delta$ \;
            $\Omega_F \leftarrow \Omega_F \cup \{ y_i \}$ \;
            $L \leftarrow L \backslash \{ y_i \}$ \;
    }
}

\Return{$\Omega_F$}

\caption{combinatorial ranking \label{alg:matheuristic}}
\end{algorithm}

Now, the complete \emph{math}euristic algorithm is defined, and we summarize it in the following three steps:
\begin{description}
\item[1.] enumerate the initial set of individual assignments $\Omega$ from the problem instance;
\item[2.] execute Algorithm \ref{alg:matheuristic} to define $\Omega_F \subset \Omega$, corresponding to a subset of individual assignments more likely to build a good solution (according to the combinatorial bound);
\item[3.] solve the IP problem $\eqref{gspp:obj}$ corresponding only to the reduced model over $\Omega_F$.
\end{description}

We remark that, since the matheuristic approach gives up on the optimality certificate, it is not imperative to reach a null duality gap in the resulting model to end the algorithm.
Such a strategy might be interesting in the event of an application with challenging runtime requirements.
Obviously, even if one is able to find a best solution to the reduced model, its optimality in the actual problem instance is by no means implied.

Finally, we also note that, on the one hand, the approach could be used with any lower bound, such as computing LP relaxations.
On the other, we argue that Algorithm \ref{alg:matheuristic} is useful because of the efficient combinatorial bound, as opposed to the corresponding time to solve LP relaxations in the loop starting at line \ref{code:delta}. That is, instead of solving as many LP problems as $|\Omega|$ to determine the lower bounds $\Delta(k)$ implied by fixing each variable $y_k$, we reduce the computational effort to ($|\Omega|$ iterations of) the construction of graph $G_{2,k}$ and the solution of a weighted matching problem on a graph with just $|T|$ vertices.

In Section \ref{sec:results} we present the results of a successful computational experience with this algorithm, using benchmark instances from a maritime logistics application. Concerning the execution time matter, even in the most difficult set of instances, solutions within 4\% of the optimal value are provided, with an average of 30\% of the time required by methods recently introduced in the literature.
On the other hand, for the sake of illustration, even if it took 1 second to build and solve the corresponding LP relaxation model, a medium-size instance with 100.000 assignments would already take over a day to compute $\Omega_F$.

\section{Sample applications}
\label{sec:applications}

In this section we aim to give straightforward examples from the literature, in which relevant applications are formulated as different generalizations of a set partitioning problem. In each case, we give a brief problem definition, transcribe its IP formulation from the literature, and show at intuition level that they are amenable to the combinatorial constructions from Sections \ref{sec:relaxations} and \ref{sec:algorithms}.

The job scheduling problem in Section \ref{sec:app:scheduling} is in direct correspondence with the GSPP structure that we present in the Introduction.
As for the applications in crew disruption management (Section \ref{sec:app:trains}) and port logistics (Section \ref{sec:app:bacap}), we highlight the interesting possibility to translate application features into small extensions of the matching relaxation.

\subsection{Job scheduling in parallel machines}
\label{sec:app:scheduling}

\cite{LallaRuizScheduling2016} concern the following problem, denoted Parallel Machine Scheduling with Step Deteriorating Jobs.
Suppose that $m$ identical, parallel machines are available in a planning horizon $h$, and that $n$ jobs are to be scheduled, each consuming a processing time varying among two constants, as follows. If the assigned starting time is on a deteriorating date $d_i$ or before, the job takes a base processing time $a_i$; if it starts after that date, it requires processing time $a_i + b_i$. For each job $i = 1, \ldots, n$, constants $a_i$, $b_i$ and $d_i$ are input parameters.
The objective function is to minimize the sum of the completion times of all jobs.

The base formulation used by those authors corresponds precisely to the GSPP structure in $\eqref{gspp:obj}-\eqref{gspp:yLB}$, as we describe next.
The set of columns is also denoted by $\Omega$.
Each column corresponds to a variable $x_w$, representing a feasible assignment to an individual job. Thus, any assignment $w$ in this problem encodes: the particular machine in which it is processed, the starting and finishing times, and a fixed cost $c_w$. Any selection of individual assignments to different jobs is compatible if no machine is used at a same time slot by more than one job.

Lalla-Ruiz and Vo{\ss} also define the index set $P$ such that $|P| = m \times h$, corresponding to machine and time slot combinations; this corresponds to $k=2$ resources in our base structure.
The authors use binary coefficient matrices $A \in \mathbb{B}^{n \times |\Omega|}$ and $B \in \mathbb{B}^{|P|\times|\Omega|}$ defined likewise $\eqref{gspp:partition}$ and $\eqref{gspp:pack}$ to express partitioning and packing constraints.
Then, \cite{LallaRuizScheduling2016} present the following IP formulation:
\begin{alignat}{2}
\min & \ \ \sum_{w \in \Omega}{c_w x_w}  \label{sched:obj} \\
\text{subject to:} & \nonumber \\
& \sum_{w \in \Omega}{A_{iw} x_w = 1}  &  \hspace{2cm} i = 1, \ldots, n    \label{sched:partition} \\
& \sum_{w \in \Omega}{B_{pw} x_w \leq 1}  &  \forall p \in P   \label{sched:pack} \\
& x_{w} \in \{0, 1\}  &  \forall w \in \Omega \label{sched:xBIN}
\end{alignat}


We thus conclude this is trivially equivalent to the GSPP structure under investigation in this paper.
Note the 1:1 correspondence between variables $x_{w}$ above and $y_j$ in the base GSPP, between entries $(a_{iw})$ in $\eqref{sched:partition}$ and $(a_{ij})$ in $\eqref{gspp:partition}$, and between $(b_{pw})$ in $\eqref{sched:pack}$ and $(b_{rj})$ in $\eqref{gspp:pack}$.
Therefore, all among the combinatorial constructions of graph $G_2$ and the bound $LB_2$ (Theorem \ref{theorem:lb2}), the variable probing test (Proposition \ref{prop:probing}), and the matheuristic in Algorithm \ref{alg:matheuristic}, carry over to this problem as presented before.

The next applications are more \emph{interesting}, in the sense that the base GSPP structure provides only a relaxation of the formulated problem, making space for application-specific tweaks. Nevertheless, we highlight the fact that the base structure can, as it is, model such relevant topics as the job scheduling problem above.

\subsection{Crew disruption management}
\label{sec:app:trains}

\cite{Rezanova2010} investigate the train driver recovery problem, which aims to find the best assignment of tasks to replacement duties for train drivers, when the railway operator has to recover from a disruption. 
In the occasion of internal or external failures (\emph{e.g.} due to track conditions, accidents, or passanger delays), such that the slack time built into the timetable is not enough to restore the original plan, a recovery mission with the re-routing or cancelling of trains is performed. Dealing with the propagation of disruptions within the schedule makes the problem rather challenging for the operator.

The application itself includes a range of details which is beyond the scope of our discussion. 
While we limit the description below to the optimization of recovery duties, the elegant work of Rezanova and Ryan concerns several auxiliary issues around that central matter, and they evaluate their contributions in the context of real-life data from a Danish railway operator.

We start with a so-called \emph{disruption neighbourhood}: let $K$ denote a subset of train drivers whose duties include at least one disrupted train task within a given recovery period (which is our planning horizon). The remaining drivers keep their former duties, and are not included in the model. Also let $N$ denote the set of all tasks originally assigned to drivers in $K$ during the period, while $P^k$ denotes the set of acceptable recovery duties for each driver $k \in K$. In this model, each individual assignment (such a recovery duty $p \in P^k$) encodes a subset of tasks in $N$ and a fixed cost $c_p^k$, corresponding to a measure of its \emph{unattractiveness} to driver $k$.

The goal of the train driver recovery problem is to find a selection of individual assignments to each driver in $K$ (\emph{i.e.} a feasible recovery duty) of least total cost, such that each task in $N$ is covered exactly once.
To this end, Rezanova and Ryan use binary decision variables $x_p^k$, set to one iff duty $p \in P^k$ is chosen.
They also define a coefficient matrix analogous to $B$ in the packing constraints $\eqref{gspp:pack}$ of the GSPP structure:
let $A$, with $|N|$ lines and a column for each duty of each driver, be such that $a_{ip}^k$ is one if task $i$ is covered by duty $p \in P^k$, and zero otherwise.
Now we can repeat their formulation, in $\eqref{trains:obj}-\eqref{trains:xBIN}$ below.
\begin{alignat}{2}
\min & \ \ \sum_{k \in K} \sum_{p \in P^k} c_p^k x_p^k  \label{trains:obj} \\
\text{subject to:} & \nonumber \\
& \sum_{p \in P^k} x_p^k = 1  &  \hspace{2cm} \forall k \in K    \label{trains:partition}\\
& \sum_{k \in K} \sum_{p \in P^k} a_{ip}^k x_p^k = 1  &  \forall i \in N   \label{trains:pack} \\
& x_p^k \in \{0, 1\}  &  \forall p \in P^k, \forall k \in K \label{trains:xBIN}
\end{alignat}

Note that variables $y_j$ in the base GSPP (\emph{i.e.} the complete collection of assignments) correspond to expanding $x_p^k$ for all $k \in K$. Therefore, summing $x_p^k$ over the different $P^k$ in constraints $\eqref{trains:partition}$ corresponds to the sum over all the products $a_{ij} y_j$ in $\eqref{gspp:partition}$, where coefficients $(a_{ij})$ \emph{select} only the assignments to task $i$ in each equality.
Analogously, the double summation over each $k \in K$ and each $p \in P^k$ in $\eqref{trains:pack}$ corresponds to the sum of all assignments $y_j \in \Omega$ in $\eqref{gspp:pack}$, and the tasks in the driver recovery problem (elements in $N$) consist of the single \emph{resource} to be distributed among the assignments (\emph{i.e.} the analougous of set $R$ in GSPP is one-dimensional in this application).

Now, the transformation of this formulation into the GSPP structure we consider allows for two possible relaxations. In each case, we describe what is relaxed, what corresponds to the packing constraints, and compare the combinatorial constructions and bounds.

A first, more natural, approach would be to relax equality constraints for tasks, allowing arbitrary $i \in N$ to remain uncovered in the solution. That is, we replace $\eqref{trains:pack}$ by:
\begin{equation}
\sum_{k \in K} \sum_{p \in P^k} a_{ip}^k x_p^k \ \leq \ 1   \hspace{2cm}   \forall i \in N
\label{trains:relax1}
\end{equation}
Then, formulation $\eqref{trains:obj},\eqref{trains:partition},\eqref{trains:xBIN},\eqref{trains:relax1}$ is a relaxation of the one by Rezanova and Ryan, and it can also be viewed as an instance of the base GSPP structure with a single shared resource.
In this case, the new, combinatorial relaxation yielding the bound in Theorem \ref{theorem:lb2} consists of solving the weighted matching problem in a complete graph $G_2$, with a vertex for each driver, such that the cost of an edge $(i,j)$ corresponds to the cheapest, non-overlapping duties for the drivers $i$ and $j$. Here, \emph{non-overlapping} means only that no task is included in both duties of that pair of drivers. 
The bound thus sums up to pairing up drivers and finding cheapest pairwise compatible combinations; some tasks will likely remain uncovered, others might be included in duties of two or more (unpaired) drivers.

Considering the real world instances solved in the work of \cite{Rezanova2010}, we verify that the number of drivers $|K|$ and the number of tasks $|N|$ have the same order of magnitude.
Since there is a collection of possible duties \emph{for each} driver $k=1,\ldots,|K|$, and matrix $A$ includes a column corresponding to each duty, $A$ should have many more columns than lines.
Therefore, the improvement of the bound described above, compared to simply removing all constraints in $\eqref{trains:pack}$, could be negligible.
Interestingly, since the formulation has two partitioning constraints, one can conceive an alternative approach, relaxing instead the equality in constraints $\eqref{trains:partition}$ into:
\begin{equation}
\sum_{p \in P^k} x_p^k \leq 1   \hspace{2cm}   \forall k \in K  
\label{trains:relax2}
\end{equation}

We therefore regard the integer program $\eqref{trains:obj},\eqref{trains:pack},\eqref{trains:xBIN},\eqref{trains:relax2}$ as a new relaxation, which requires that all tasks are covered exactly once, while allowing some drivers to remain idle.
From this point, the new matching relaxation, attaining the bound from Theorem \ref{theorem:lb2} for this application, builds on a complete graph, with a vertex for each task. Now, the weight of an edge $(i,j)$ is determined by the cheapest pair of recovery duties covering both taks $i$ and $j$ exactly once, in which no driver is assigned two different duties.
That is, the compatibility oracle would return the cost of either: (i) a single duty for the same driver, including both tasks, or (ii) two duties, for different drivers, each covering one of the tasks but not the other.
We remark that imparting more application details on the oracle decision for pairwise-compatible assignments (\emph{i.e.} possibly tightening the gap towards an overall feasible solution) could increase the edge weights and, consequently, strengthen the matching bound.

\subsection{Port logistics}
\label{sec:app:bacap}

Our last sample application is also the subject matter of our preliminary computational evaluation, in the next section.
The Berth Allocation and Quay Crane Assignment Problem (BACAP) aims to allocate a berthing time, a position in the quay, and a number of quay cranes (QCs) for arriving vessels in a seaport container terminal. Feasible assignments in the BACAP need to fulfil requirements on desired berthing period and position, and an agreement on the use of QCs.
General reviews and a taxonomy to compare related work can be found in the surveys of \cite{Stahlbock2008}, \cite{survey2010}, and \cite{survey2015}.

We follow the presentation by \cite{iris2015}, who give a precise description of the application characteristics, 
and a clear classification of related  literature.
Let $V$ be the set of vessels, $T$ be the set of time slots in the planning horizon, $L$ be the set of berthing positions in the quay, and $K$ be the number of available QCs.
Also define the set of berthing time/position combinations $P = T \times L$; this corresponds to two shared resources in our base GSPP structure.
They define the set $\Omega$ as we have used throughout this paper: the complete set of feasible individual allocations. In this case, each element $j \in \Omega$ encondes a suitable space in the quay, time slots and a number of cranes to serve a given vessel, besides its cost $c_j$. This cost, in turn, depends on QC usage, deviations from parameters on the desired position on the quay and expected starting and finishing times for the service.

Binary decision variables $y \in \mathbb{B}^{|\Omega|}$ indicate which individual assignments are used in the solution.
Finally, the coefficient matrices are in accordance with those in our description of the base GSPP, as we describe next.
$A \in \mathbb{B}^{|V|\times|\Omega|}$ associates each column with a single vessel: $a_{ij}$ is equal to one if column $j$ refers to an assignment for vessel $i$; otherwise, it is equal to zero.
$B \in \mathbb{B}^{|P|\times|\Omega|}$ represents berths as combinations of time intervals and quay positions: $b_{pj}$ is one iff the given pair of (time, space) positions corresponding to $p \in P$ is used in the assignment $y_j$.
An element of $Q \in \mathbb{Z}^{|T|\times|\Omega|}$ determines how many QCs are used by $y_j$ in time period $t$. Then, the BACAP formulation described by \cite{iris2015} corresponds to:
\begin{alignat}{2}
\min & \ \ \sum_{j \in \Omega} c_j y_j  \label{bacap:obj} \\
\text{subject to:} & \nonumber \\
& \sum_{j \in \Omega} a_{ij} y_j = 1  &  \hspace{2cm} \forall i \in V    \label{bacap:partition}\\
& \sum_{j \in \Omega} b_{pj} y_j \leq 1 &  \forall p \in P   \label{bacap:pack}\\
& \sum_{j \in \Omega} q_{tj} y_j \leq K  & \forall t \in T   \label{bacap:qcs}\\
& y_j \in \{0, 1\}  &  \forall j \in \Omega \label{bacap:yBIN}
\end{alignat}

Set partition constraints $\eqref{bacap:partition}$ ensure that all vessels are served by exactly one assignment, while set packing in $\eqref{bacap:pack}$ forbid overlapping of vessel assignments in each single time/space slot.
Inequalities $\eqref{bacap:qcs}$ guarantee that QCs availability in the terminal is respected.


It is clear, then, that one need only relax the latter set of inequalities $\eqref{bacap:qcs}$ to view the formulation by Iris \emph{et al.} as an instance of the base GSPP structure we concern, with a two-dimensional analougous of the set $R$ of resource-tuples in this application (time slots and quay space combinations); the remainder of the above formulation is in 1:1 correspondence with the base GSPP.
The matching relaxation from Theorem \ref{theorem:lb2} can thus be determined on a complete graph, with a vertex for each vessel, where the cost of an edge amounts to the weight of a cheapest, compatible pair of assignments to the corresponding vessels.

Finally, we remark that two assignments for different vessels in the BACAP are denoted \emph{compatible} if they have no overlap in berthing time and space. Equivalently, representing the two assignments in a Cartesian plane (with time and space coordinates), they are compatible iff the corresponding rectangles do not intersect each other.
We can further tighten the definition of compatibility in this application by limiting the combined number of cranes used by two given assignments to the maximum available in the quay.
As in the case of the previous application, this stronger compatibility criterion might rule out pairs of assignments and increase edge weights in $G_2$, possibly strengthening the matching bound.

\section{A computational case study}
\label{sec:results}


The goals of the computational evaluation we present are twofold.
First, we want to compare the linear programming and the matching relaxations to evaluate the trade-off between the strength of the bound and the time to compute it.
Second, we seek to assess the quality of the solutions provided by the matheuristic we introduced.
Toward these ends, we have implemented a series of algorithms concerning the last application we described, the Berth Allocation and Quay Crane Assignment Problem (BACAP).
We consider the same benchmark instances used by \cite{meisel2009} and \cite{iris2015} to evaluate the efficiency of the proposed algorithms.
There are thirty instances of three different sizes, ten of each: 20 (small), 30 (medium) and 40 (large) vessels.

All algorithms were implemented in C++ using Gurobi solver version 6.5.
The number of execution threads allowed for Gurobi to solve each IP model was set to 2 in the case of instances with 20 or 30 vessels, while a single thread is allowed to solve instances with 40 vessels. Since the Gurobi solver needs copies of the complete model for each execution thread, we verified improved performance using this setting because no virtual memory is needed.
To compute maximum weighted matchings, with the blossom shrinking algorithm by \cite{Edmonds1965PTF}, we used the efficient implementation available in the open source Library for Efficient Modeling and Optimization in Networks (LEMON) \citep{lemon2011}. The time complexity of that implementation is $O(mn\log n)$ in the worst case, where $n$ is the number of graph vertices and $m$ is the number of edges.

All experiments were run in a machine with an Intel Core i7 4790K (4.00 GHz) CPU and 16GB of RAM.
The solver runtime used in the matheuristic experiments was limited to 1800 seconds.
The experiments corresponding to the solution approach proposed by \cite{iris2015} were not time limited.
It is important to highlight that all results we present concerning the work by those authors were evaluated with our own implementation of their algorithms. Moreover, the numbers concerning their results may have slight variations when compared to the original work, which we have concluded to be explained by numeric precision matters.


First, we compare the lower bounds and the runtime to compute the combinatorial relaxation, against the LP relaxation. Table \ref{tab:lb} presents the corresponding bounds in columns 3 and 4. Column 5 indicates the percentual difference of the combinatorial bound with respect to the LP relaxation one. Columns 6 and 7 present the execution time required by each method.



\begin{table}[ht!]
\centering
\setstretch{1.35}
\caption{Results attained with the combinatorial and the linear programming relaxations}
\label{tab:lb}
\begin{threeparttable}
\scriptsize
\begin{tabular}{ccccccccc}
\toprule
\multicolumn{2}{c}{Instance}  &&  \multicolumn{3}{c}{Lower Bounds}  &&  \multicolumn{2}{c}{Time (s)}  \\
\cmidrule{1-2}  \cmidrule{4-6}  \cmidrule{8-9}
\multirow{2}{0.5cm}{$|V|$} & \multirow{2}{0.5cm}{ID} && LP & Matching & Difference && LP & Matching \\
   &   &&  relaxation  &  relaxation  &  (\%)  &&  relaxation &  relaxation \\
\midrule
\multirow{10}{0.5cm}{20}    &   1   &&   885.0   &   698.0 &   21.1\%    &&   5.61    &   0.01    \\
                            &   2   &&   562.0   &   562.0 & \ 0.0\%     &&   0.16    &   0.00    \\
                            &   3   &&   816.5   &   646.0 & 20.9\%      &&   6.69    &   0.02    \\
                            &   4   &&   762.0   &   620.0 & 18.6\%      &&   7.37    &   0.02    \\
                            &   5   &&   592.0   &   516.0 & 12.8\%      &&   1.62    &   0.01    \\
                            &   6   &&   592.0   &   592.0 & \ 0.0\%     &&   1.23    &   0.00    \\
                            &   7   &&   722.0   &   646.0 & 10.5\%      &&   2.77    &   0.00    \\
                            &   8   &&   582.0   &   532.0 & 8.6\%       &&   4.02    &   0.00    \\
                            &   9   &&   782.0   &   620.0 & 20.7\%      &&   6.47    &   0.01    \\
                            &   10  &&   930.7   &   700.0 & 24.8\%      &&   9.34    &   0.02    \\
                            &       &&           &   $\overline{\text{LBgap}}_{I_{{20}}}$  & 13.8\%  &           &           \\
\midrule
\multirow{10}{0.5cm}{30}    &   11  &&   1,408.8 &   \ 922.0 &   34.6\%   &&   29.98   &   0.06    \\
                            &   12  &&   891.3   &   \ 800.0 &   10.2\%   &&   \ 5.23  &   0.01    \\
                            &   13  &&   1,091.7 &   \ 894.0 &   18.1\%   &&   \ 7.64  &   0.02    \\
                            &   14  &&   1,036.5 &   \ 880.0 &   15.1\%   &&   12.19   &   0.02    \\
                            &   15  &&   1,600.1 &   1,046.0 &   34.6\%   &&   36.58   &   0.09    \\
                            &   16  &&   1,137.0 &   1,008.0 &   11.3\%   &&   \ 7.56  &   0.01    \\
                            &   17  &&   1,084.0 &   \ 894.0 &   17.5\%   &&   \ 9.07  &   0.02    \\
                            &   18  &&   1,245.0 &   \ 860.0 &   30.9\%   &&   31.71   &   0.07    \\
                            &   19  &&   1,705.0 &   1,052.0 &   38.3\%   &&   44.07   &   0.08    \\
                            &   20  &&   1,354.4 &   1,008.0 &   25.6\%   &&   22.13   &   0.03    \\
                            &       &&           &   $\overline{\text{LBgap}}_{I_{{30}}}$  & 23.6\%     &&           &           \\
\midrule                                            
\multirow{10}{0.5cm}{40}    &   21  &&   2,058.8 &   1,150.0   &   44.1\%   &&   94.29   &   0.22    \\
                            &   22  &&   1,680.0 &   1,288.0   &   23.3\%   &&   \ 36.68 &   0.12    \\
                            &   23  &&   2,380.7 &   1,250.0   &   47.5\%   &&   102.59  &   0.28    \\
                            &   24  &&   2,727.0 &   1,544.0   &   43.4\%   &&   118.38  &   0.38    \\
                            &   25  &&   1,559.9 &   1,102.0   &   29.4\%   &&   \ 44.00 &   0.13    \\
                            &   26  &&   2,364.1 &   1,294.0   &   45.3\%   &&   \ 97.34 &   0.31    \\
                            &   27  &&   1,965.8 &   1,222.0   &   37.8\%   &&   \ 58.83 &   0.18    \\
                            &   28  &&   2,533.3 &   1,412.0   &   44.3\%   &&   122.40  &   0.41    \\
                            &   29  &&   2,071.8 &   1,404.0   &   32.2\%   &&   \ 65.54 &   0.16    \\
                            &   30  &&   1,872.5 &   1,240.0   &   33.8\%   &&   \ 73.21 &   0.21    \\
                            &       &&           &   $\overline{\text{LBgap}}_{I_{{40}}}$  & 38.1\%     &&           &           \\
\bottomrule
\end{tabular}
\begin{tablenotes}
      \scriptsize
      \item $\overline{\text{LBgap}}_{I_{{|V|}}}$ stands for the arithmetic mean of the percentual difference in the bounds for instances with $|V|$ vessels.
\end{tablenotes}
\end{threeparttable}
\end{table}

The combinatorial bound matches the LP relaxation value for two instances, but it is consistently weaker and the difference also grows with the input size.
Nevertheless, while building the model and solving its LP relaxation consumes significant time, building the graphs representing compatible assignment and solving the corresponding weighted matching problem is performed very fast, in comparison.

One can thus argue that, when choosing a lower bound to use in the variable reduction technique or in the matheuristic algorithm, the computational performance of the matching relaxation is the crucial factor.
Since a lower bound must be computed after probing each variable, and even the smallest instances have tens of thousands of variables, the runtime of the LP relaxation becomes prohibitive in this context.

Next, we executed different experiments to evaluate the matheuristic algorithm.
We stress again that the matheuristic approach waives the optimality certificate of a solution in the sake of a reduced computation time.
In the following, we verify this outcome and assess the strength of the solutions attained, in comparison with the exact approach in the literature.

\begin{sloppypar}
Several matheuristic parameter combinations were tested, with $\sigma \in \{0.0, 0.1, 0.2, 0.3\}$ and $\mu \in \{500, 1000, 1500, 2000\}$.
Preliminary evaluations with $\mu = 0$ and different choices for $\sigma$ had a poor performance, 
leading to such an extreme reduction that the resulting model was infeasible for at least one instance, \emph{i.e.} the corresponding polyhedron does not include a single integer point.
This is important to justify the final step in our algorithm (requiring at least $\mu$ variables corresponding to each vessel).
The four parameter combinations with $\sigma \in \{0.0, 0.1\}$ and $\mu \in \{500, 1000\}$ also led to an infeasible model for at least one instance.
The remaining combinations always found an integer feasible solution.
\end{sloppypar}

We report in Table~\ref{tab:quality} the results with the configuration $(\sigma=0.1, \mu=2000)$, which yields the least average GAP$_{\text{OPT}}$, between the best primal solutions and the known optima.
For each instance, we present 
the optimal solution value $z$, 
the percentage of variables remaining after the matheuristic filter ($\Omega_F$), 
followed by the cost $\bar{z}$ of the best solution found through the matheuristic and its GAP$_{\text{OPT}}$ (to the optimal value $z$).
In the next columns, we present the partial runtime $P_C$ of applying our combinatorial reduction technique from Proposition \ref{prop:probing}, and the partial runtime $P_I$ of applying the one by \cite{iris2015}. Then follows the respective total runtimes $T_{C}$ and $T_{I}$ (\emph{i.e.} the sum of the times for preprocessing and for solving the mathematical model). 
The last column presents the runtime improvement of our proposal  compared to the one from the literature: $E = (\nicefrac{T_C}{T_I}) \times 100$.

\begin{table}[ht!]
\setstretch{1.35}
\caption{Results concerning quality of solutions and runtime}
\label{tab:quality}
\makebox[\textwidth][c]{
\begin{threeparttable}
\scriptsize
\begin{tabular}{ccrrrcrrrrrrrr}
\toprule
\multicolumn{2}{c}{Instance}	& 	& \multicolumn{3}{c}{Matheuristic results} & & \multicolumn{5}{c}{Runtime} & & \multicolumn{1}{c}{Efficiency}	  \\

\cmidrule{1-2} \cmidrule{4-6} \cmidrule{8-12} \cmidrule{14-14}
	&	& 	& & 	& &	&	\multicolumn{2}{c}{Matheuristic}	& &	\multicolumn{2}{c}{\cite{iris2015} }	& 	&	\\

 \cmidrule{8-9}  \cmidrule{11-12} 
$|V|$	&	ID	& \multicolumn{1}{c}{$z$}	& 	\multicolumn{1}{c}{$\Omega_{F} (\%)$} & \multicolumn{1}{c}{$\bar{z}$}	&	$\text{Gap}_{\text{OPT}} ($\%$)$ &	&	$P_{C}$ (s)	&	$T_{C}$ (s)	& &	$P_{I}$ (s)	&	$T_{I}$ (s) 	& 	&	$E (\%)$	\\
\midrule
\multirow{10}{0.5cm}{20}	&	1	&	89.00	&	60.84	&	89.00	&	0.00	&	&	32.44	&	\textbf{56.65}	&	&	14.84	&	64.28	&	&	88.13	\\
							&	2	&	56.20	&	100.00	&	56.20	&	0.00	&	&	1.01	&	1.09	&	&	0.09	&	\textbf{0.16}	&	&	681.25	\\
							&	3	&	85.70	&	59.09	&	85.70	&	0.00	&	&	9.63	&	\textbf{73.19}	&	&	0.79	&	114.73	&	&	63.79	\\
							&	4	&	81.80	&	47.86	&	81.80	&	0.00	&	&	8.58	&	\textbf{39.24}	&	&	2.22	&	60.22	&	&	65.16	\\
							&	5	&	59.20	&	99.13	&	59.20	&	0.00	&	&	2.62	&	7.99	&	&	0.16	&	\textbf{7.59}	&	&	105.27	\\
							&	6	&	59.20	&	100.00	&	59.20	&	0.00	&	&	1.53	&	6.73	&	&	0.16	&	\textbf{5.09}	&	&	132.22	\\
							&	7	&	75.20	&	94.32	&	75.20	&	0.00	&	&	6.42	&	21.33	&	&	2.56	&	\textbf{19.92}	&	&	107.08	\\
							&	8	&	61.40	&	73.20	&	61.40	&	0.00	&	&	5.39	&	22.10	&	&	0.95	&	\textbf{21.41}	&	&	103.22	\\
							&	9	&	79.00	&	52.22	&	79.00	&	0.00	&	&	16.37	&	\textbf{45.45}	&	&	3.77	&	54.32	&	&	83.67	\\
							&	10	&	101.00	&	45.56	&	101.00	&	0.00	&	&	26.05	&	\textbf{76.27}	&	&	3.13	&	214.43	&	&	35.57	\\
	& & & & & & & & & & & $\overline{E_{20}}$ &   & 101.34\\
\midrule																								
\multirow{10}{0.5cm}{30}	&	11	&	143.20	&	31.43	&	143.20	&	0.00	&	&	60.66	&	\textbf{186.16}	&	&	7.64	&	421.91	&	&	44.12	\\
							&	12	&	92.00	&	87.55	&	92.00	&	0.00	&	&	13.59	&	43.93	&	&	0.85	&	\textbf{43.67}	&	&	100.60	\\
							&	13	&	110.00	&	66.45	&	110.00	&	0.00	&	&	24.53	&	\textbf{64.00}	&	&	3.01	&	85.90	&	&	74.51	\\
							&	14	&	107.40	&	50.85	&	107.40	&	0.00	&	&	28.50	&	\textbf{111.48}	&	&	6.73	&	164.24	&	&	67.88	\\
							&	15	&	168.40	&	26.53	&	168.40	&	0.00	&	&	70.46	&	\textbf{236.08}	&	&	10.64	&	665.77	&	&	35.46	\\
							&	16	&	121.60	&	72.09	&	121.60	&	0.00	&	&	57.76	&	\textbf{144.80}	&	&	22.03	&	153.33	&	&	94.44	\\
							&	17	&	109.40	&	61.39	&	109.40	&	0.00	&	&	28.97	&	\textbf{73.98}	&	&	4.30	&	85.56	&	&	86.47	\\
							&	18	&	135.00	&	28.24	&	135.00	&	0.00	&	&	56.29	&	\textbf{253.90}	&	&	10.72	&	651.13	&	&	38.99	\\
							&	19	&	176.20	&	23.18	&	176.20	&	0.00	&	&	86.51	&	\textbf{252.40}	&	&	27.41	&	737.74	&	&	34.21	\\
							&	20	&	139.80	&	38.27	&	139.80	&	0.00	&	&	70.84	&	\textbf{204.72}	&	&	28.60	&	346.50	&	&	59.08	\\
		& & & & & & & & & & & $\overline{E_{30}}$ &   & 59.01\\
\midrule																								
\multirow{10}{0.5cm}{40}	&	21	&	246.80	&	16.20	&	252.60	&	2.30	&	&	235.29	&	\textbf{2035.52}	&	&	20.90	&	9237.00	&	&	22.04	\\
							&	22	&	178.40	&	30.14	&	178.40	&	0.00	&	&	171.43	&	\textbf{449.33}	&	&	45.42	&	694.91	&	&	64.66	\\
							&	23	&	266.30	&	15.39	&	266.30	&	0.00	&	&	322.13	&	\textbf{1771.94}	&	&	131.24	&	14516.88	&	&	12.21	\\
							&	24	&	307.00	&	15.39	&	318.20	&	3.52	&	&	1788.03	&	\textbf{3588.11}	&	&	538.36	&	11071.74	&	&	32.41	\\
							&	25	&	164.60	&	28.63	&	164.60	&	0.00	&	&	105.97	&	\textbf{457.75}	&	&	8.21	&	1049.86	&	&	43.60	\\
							&	26	&	258.20	&	16.31	&	261.90	&	1.41	&	&	274.46	&	\textbf{2074.55}	&	&	89.55	&	4337.52	&	&	47.83	\\
							&	27	&	205.40	&	19.92	&	205.40	&	0.00	&	&	171.92	&	\textbf{454.28}	&	&	17.81	&	1520.09	&	&	29.89	\\
							&	28	&	294.20	&	16.04	&	301.30	&	2.36	&	&	518.37	&	\textbf{2318.46}	&	&	66.65	&	42988.26	&	&	5.39	\\
							&	29	&	227.60	&	19.96	&	229.70	&	0.91	&	&	680.40	&	\textbf{2480.52}	&	&	567.79	&	2563.04	&	&	96.78	\\
							&	30	&	210.10	&	18.57	&	210.10	&	0.00	&	&	223.22	&	\textbf{1470.65}	&	&	46.59	&	4354.89	&	&	33.77	\\
	& & & & & & & & & & & $\overline{E_{40}}$ &   & 30.14\\
\bottomrule
\end{tabular}

	\begin{tablenotes}
      \scriptsize
      \item $\Omega_{F} = (\nicefrac{\Omega_M}{\Omega_1}) \times 100$, where $\Omega_M$ denotes the number of variables remaining after the matheuristic filter and $\Omega_1$ denotes the number of variables after the enumeration and preprocessing phases; $E = (\nicefrac{T_C}{T_I}) \times 100$; $\overline{E_{|V|}}$ corresponds to the geometric mean of $E$ for instances with $|V|$ vessels.
    \end{tablenotes}
\end{threeparttable}
}
\end{table}

It can be seen from these results that the proposed methodology is able to find the known optimal solution in 83\% of instances, while the gap is below 4\% for those solutions which are not optimal.
Note that the runtime efficiency of our methodology is inferior to the literature in half of the instances with 20 vessels and in one of the medium instances.
Nevertheless, these results do not have a significant impact because the corresponding execution times are at most 44 seconds, and the time difference between the methods does not exceed 2 seconds. 
For example, while instance~2 is solved instantaneously by the technique from the literature, it takes $1.09$ seconds using the matheuristic. Therefore, our proposal looses on the average efficiency for small instances.

On the other hand, observing the results for medium and large instances in the benchmark, it is clear that the matheuristic times are significantly smaller than those from the exact method presented by \cite{iris2015}.
The known optimal solution is found in all instances with 30 vessels, spending on average 59\% of the time spent by the exact method.
For more difficult instances, with 40 vessels, the known optimal solution is attained in half of the cases, while very good solutions can be achieved using about 30\% of the time consumed by the baseline method.

Considering instance~23, for example, the optimal solution is obtained 8 times faster by the matheuristic algorithm. In instance~28, the matheuristic consumes less than 6\% of literature algorithm runtime to find a solution within less than 3\% of the optimal value.
These results suggest that ranking the variables with the combinatorial bound is an effective selection criterion: the model size decreases drastically, the runtime reduces consistently in instances with more challenging dimensions, while keeping solutions within 4\% of the optimum.

\section{Concluding remarks}
\label{sec:conclusion}
We investigate in this work two discrete relaxations for the class of combinatorial optimization problems which can be described by a set partitioning formulation under packing constraints.
We introduce two lower bounds to this class of generalized set partitioning problems (GSPP), based on computing weighted matchings on suitable graphs. We prove that one of these bounds is stronger than the other and, on the algorithmic side, use the relaxations as the basis of a variable reduction technique and of a matheuristic method.

We seek to claim attention to the relevance of the GSPP structure as interesting in its own right, and indicate that it can be used to model different applications in the recent literature. For instance, we apply that base GSPP and our constructions to three sample applications concerning job scheduling, crew disruption management, and maritime logistics. In the particular case of the latter application, we have conducted a computational case study, and the \emph{matching relaxation setup} has proven useful to reach the known optimal solutions on 83\% of the benchmark instances of the problem; otherwise, the solutions were within 4\% of the optimal value, while using an average of 30\% of the time required by methods recently introduced in the literature.

It is worth remarking that the purely combinatorial constructions that we devised may be extended to include application-specific features, as we illustrate with two possible relaxations in the case of the train driver recovery problem, and with the stronger compatibility criterion depicted in the port logistics problem.

There are several tracks for future work on this problem. First, the subproblem of determining the best compatible assignments for two tasks is critical in the computation time of the combinatorial bound and, therefore, the variable reduction technique. One could investigate whether specialized data structures would allow designing a more efficient implementation. We also suggest the statistical study of matheuristic variants, using different criteria for choosing its parameters.
Finally, it could be interesting to study alternative relaxations of the base GSPP structure, \emph{e.g.} comparing combinatorial and Lagrangean bounds, or even extending the relax-and-cut approach of \cite{CavalcanteCidAbilio2008} to the GSPP we explore.

\bibliographystyle{plainnat-no-doi}

\begin{thebibliography}{37}
\providecommand{\natexlab}[1]{#1}
\providecommand{\url}[1]{\texttt{#1}}
\expandafter\ifx\csname urlstyle\endcsname\relax
  \providecommand{\doi}[1]{doi: #1}\else
  \providecommand{\doi}{doi: \begingroup \urlstyle{rm}\Url}\fi

\bibitem[Ali and Thiagarajan(1989)]{ali1989network}
A.~I. Ali and H.~Thiagarajan.
\newblock A network relaxation based enumeration algorithm for set
  partitioning.
\newblock \emph{European Journal of Operational Research}, 38\penalty0
  (1):\penalty0 76--85, 1989.
\newblock URL \url{http://dx.doi.org/10.1016/0377-2217(89)90471-2}.

\bibitem[Atamt{\"u}rk et~al.(2000)Atamt{\"u}rk, Nemhauser, and
  Savelsbergh]{atamturk2000conflict}
A.~Atamt{\"u}rk, G.~L. Nemhauser, and M.~W. Savelsbergh.
\newblock Conflict graphs in solving integer programming problems.
\newblock \emph{European Journal of Operational Research}, 121\penalty0
  (1):\penalty0 40--55, 2000.
\newblock URL \url{http://dx.doi.org/10.1016/S0377-2217(99)00015-6}.

\bibitem[Balas and Padberg(1976)]{balas1976set}
E.~Balas and M.~W. Padberg.
\newblock Set partitioning: A survey.
\newblock \emph{SIAM review}, 18\penalty0 (4):\penalty0 710--760, 1976.
\newblock URL \url{http://dx.doi.org/10.1137/1018115}.

\bibitem[Ball et~al.(1990)Ball, Derigs, Hilbrand, and Metz]{ball1990matching}
M.~O. Ball, U.~Derigs, C.~Hilbrand, and A.~Metz.
\newblock Matching problems with generalized upper bound side constraints.
\newblock \emph{Networks}, 20\penalty0 (6):\penalty0 703--721, 1990.
\newblock URL \url{http://dx.doi.org/10.1002/net.3230200602}.

\bibitem[Bierwirth and Meisel(2010)]{survey2010}
C.~Bierwirth and F.~Meisel.
\newblock A survey of berth allocation and quay crane scheduling problems in
  container terminals.
\newblock \emph{European Journal of Operational Research}, 202\penalty0
  (3):\penalty0 615 -- 627, 2010.
\newblock ISSN 0377-2217.
\newblock URL \url{http://dx.doi.org/10.1016/j.ejor.2009.05.031}.

\bibitem[Bierwirth and Meisel(2015)]{survey2015}
C.~Bierwirth and F.~Meisel.
\newblock A follow-up survey of berth allocation and quay crane scheduling
  problems in container terminals.
\newblock \emph{European Journal of Operational Research}, 244\penalty0
  (3):\penalty0 675 -- 689, 2015.
\newblock ISSN 0377-2217.
\newblock URL \url{http://dx.doi.org/10.1016/j.ejor.2014.12.030}.

\bibitem[Bornd{\"o}rfer(1998)]{borndorfer1998Thesis}
R.~Bornd{\"o}rfer.
\newblock \emph{Aspects of set packing, partitioning, and covering}.
\newblock PhD thesis, Doctoral dissertation, Ph. D. thesis, Technischen
  Universit{\"a}t Berlin, Germany, 1998.

\bibitem[Bornd{\"o}rfer(2004)]{2004BorndorferChapter}
R.~Bornd{\"o}rfer.
\newblock \emph{Combinatorial Packing Problems, \emph{in} The Sharpest Cut},
  chapter~3, pages 19--32.
\newblock Society for Industrial and Applied Mathematics (SIAM), 2004.
\newblock URL \url{http://locus.siam.org/doi/abs/10.1137/1.9780898718805.ch3}.

\bibitem[Bornd{\"o}rfer and Weismantel(2001)]{BorndorferWeismantel2001}
R.~Bornd{\"o}rfer and R.~Weismantel.
\newblock Discrete relaxations of combinatorial programs.
\newblock \emph{Discrete Applied Mathematics}, 112\penalty0 (1–3):\penalty0
  11 -- 26, 2001.
\newblock ISSN 0166-218X.
\newblock URL \url{http://dx.doi.org/10.1016/S0166-218X(00)00307-3}.

\bibitem[Boschetti et~al.(2008)Boschetti, Mingozzi, and
  Ricciardelli]{boschetti2008}
M.~A. Boschetti, A.~Mingozzi, and S.~Ricciardelli.
\newblock A dual ascent procedure for the set partitioning problem.
\newblock \emph{Discrete Optimization}, 5\penalty0 (4):\penalty0 735--747,
  2008.
\newblock URL \url{https://doi.org/10.1016/j.disopt.2008.06.001}.

\bibitem[Campello and Maculan(1987)]{Campello1987}
R.~E. Campello and N.~F. Maculan.
\newblock Lagrangean relaxation for a lower bound to a set partitioning problem
  with side constraints: properties and algorithms.
\newblock \emph{Discrete Applied Mathematics}, 18\penalty0 (2):\penalty0 119 --
  136, 1987.
\newblock ISSN 0166-218X.
\newblock URL \url{http://dx.doi.org/10.1016/0166-218X(87)90015-1}.

\bibitem[Cavalcante et~al.(2008)Cavalcante, de~Souza, and
  Lucena]{CavalcanteCidAbilio2008}
V.~F. Cavalcante, C.~C. de~Souza, and A.~Lucena.
\newblock A relax-and-cut algorithm for the set partitioning problem.
\newblock \emph{Computers \& Operations Research}, 35\penalty0 (6):\penalty0
  1963 -- 1981, 2008.
\newblock ISSN 0305-0548.
\newblock URL \url{http://dx.doi.org/10.1016/j.cor.2006.10.009}.

\bibitem[Chan and Yano(1992)]{chan1992}
T.~J. Chan and C.~A. Yano.
\newblock A multiplier adjustment approach for the set partitioning problem.
\newblock \emph{Operations Research}, 40\penalty0 (Supplement 1:
  Optimization):\penalty0 S40--S47, 1992.
\newblock URL \url{https://doi.org/10.1287/opre.40.1.S40}.

\bibitem[Dezs{\H{o}} et~al.(2011)Dezs{\H{o}}, J\"{u}ttner, and
  Kov\'{a}cs]{lemon2011}
B.~Dezs{\H{o}}, A.~J\"{u}ttner, and P.~Kov\'{a}cs.
\newblock {LEMON -- an Open Source C++ Graph Template Library}.
\newblock \emph{Electronic Notes in Theoretical Computer Science}, 264\penalty0
  (5):\penalty0 23 -- 45, 2011.
\newblock ISSN 1571-0661.
\newblock URL \url{http://dx.doi.org/10.1016/j.entcs.2011.06.003}.

\bibitem[Edmonds(1965)]{Edmonds1965PTF}
J.~Edmonds.
\newblock Paths, trees, and flowers.
\newblock \emph{Canadian Journal of Mathematics}, 17:\penalty0 449--467, 1965.
\newblock URL \url{http://cms.math.ca/10.4153/CJM-1965-045-4}.

\bibitem[El-Darzi and Mitra(1995)]{el1995graph}
E.~El-Darzi and G.~Mitra.
\newblock Graph theoretic relaxations of set covering and set partitioning
  problems.
\newblock \emph{European Journal of Operational Research}, 87\penalty0
  (1):\penalty0 109--121, 1995.
\newblock URL \url{http://dx.doi.org/10.1016/0377-2217(94)00115-S}.

\bibitem[Fanjul-Peyro and Ruiz(2011)]{FanjulPeyro2011}
L.~Fanjul-Peyro and R.~Ruiz.
\newblock Size-reduction heuristics for the unrelated parallel machines
  scheduling problem.
\newblock \emph{Computers \& Operations Research}, 38\penalty0 (1):\penalty0
  301 -- 309, 2011.
\newblock ISSN 0305-0548.
\newblock URL \url{http://dx.doi.org/10.1016/j.cor.2010.05.005}.
\newblock Project Management and Scheduling.

\bibitem[Fisher and Kedia(1990)]{Fisher1990}
M.~L. Fisher and P.~Kedia.
\newblock Optimal solution of set covering/partitioning problems using dual
  heuristics.
\newblock \emph{Management Science}, 36\penalty0 (6):\penalty0 674--688, 1990.
\newblock URL \url{http://dx.doi.org/10.1287/mnsc.36.6.674}.

\bibitem[Hoffman and Padberg(1993)]{hoffman1993airline}
K.~L. Hoffman and M.~Padberg.
\newblock Solving airline crew scheduling problems by branch-and-cut.
\newblock \emph{Management Science}, 39\penalty0 (6):\penalty0 657--682, 1993.
\newblock URL \url{https://doi.org/10.1287/mnsc.39.6.657}.

\bibitem[Iris et~al.(2015)Iris, Pacino, Ropke, and Larsen]{iris2015}
C.~Iris, D.~Pacino, S.~Ropke, and A.~Larsen.
\newblock Integrated berth allocation and quay crane assignment problem: Set
  partitioning models and computational results.
\newblock \emph{Transportation Research Part E: Logistics and Transportation
  Review}, 81:\penalty0 75 -- 97, 2015.
\newblock ISSN 1366-5545.
\newblock URL \url{http://dx.doi.org/10.1016/j.tre.2015.06.008}.

\bibitem[Kuo and Leung(2016)]{kuo2016}
Y.-H. Kuo and J.~M. Leung.
\newblock On the mixed set covering, packing and partitioning polytope.
\newblock \emph{Discrete Optimization}, 22, Part A:\penalty0 162 -- 182, 2016.
\newblock ISSN 1572-5286.
\newblock URL \url{http://dx.doi.org/10.1016/j.disopt.2016.05.004}.

\bibitem[Lalla-Ruiz and Vo{\ss}(2016{\natexlab{a}})]{LallaRuizBAP2016}
E.~Lalla-Ruiz and S.~Vo{\ss}.
\newblock Popmusic as a matheuristic for the berth allocation problem.
\newblock \emph{Annals of Mathematics and Artificial Intelligence}, 76\penalty0
  (1):\penalty0 173--189, 2016{\natexlab{a}}.
\newblock ISSN 1573-7470.
\newblock URL \url{http://dx.doi.org/10.1007/s10472-014-9444-4}.

\bibitem[Lalla-Ruiz and Vo{\ss}(2016{\natexlab{b}})]{LallaRuizScheduling2016}
E.~Lalla-Ruiz and S.~Vo{\ss}.
\newblock Modeling the parallel machine scheduling problem with step
  deteriorating jobs.
\newblock \emph{European Journal of Operational Research}, 255\penalty0
  (1):\penalty0 21 -- 33, 2016{\natexlab{b}}.
\newblock ISSN 0377-2217.
\newblock URL \url{http://dx.doi.org/10.1016/j.ejor.2016.04.010}.

\bibitem[Letchford(2001)]{Letchford2001}
A.~N. Letchford.
\newblock On disjunctive cuts for combinatorial optimization.
\newblock \emph{Journal of Combinatorial Optimization}, 5\penalty0
  (3):\penalty0 299--315, 2001.
\newblock ISSN 1573-2886.
\newblock URL \url{http://dx.doi.org/10.1023/A:1011493126498}.

\bibitem[Linderoth et~al.(2001)Linderoth, Lee, and Savelsbergh]{linderoth2001}
J.~T. Linderoth, E.~K. Lee, and M.~W. Savelsbergh.
\newblock A parallel, linear programming-based heuristic for large-scale set
  partitioning problems.
\newblock \emph{INFORMS Journal on Computing}, 13\penalty0 (3):\penalty0
  191--209, 2001.
\newblock URL \url{https://doi.org/10.1287/ijoc.13.3.191.12630}.

\bibitem[Maniezzo et~al.(2010)Maniezzo, St{\"u}tzle, and
  Vo{\ss}]{matheuristics2009book}
V.~Maniezzo, T.~St{\"u}tzle, and S.~Vo{\ss}.
\newblock \emph{Matheuristics: hybridizing metaheuristics and mathematical
  programming}, volume~10 of \emph{Annals of Information Systems}.
\newblock Springer US, 2010.
\newblock ISBN 978-1-4419-1305-0.
\newblock URL \url{http://dx.doi.org/10.1007/978-1-4419-1306-7}.

\bibitem[Mauri et~al.(2008)Mauri, Oliveira, and Lorena]{Mauri2008}
G.~R. Mauri, A.~C.~M. Oliveira, and L.~A.~N. Lorena.
\newblock A hybrid column generation approach for the berth allocation problem.
\newblock In J.~van Hemert and C.~Cotta, editors, \emph{Proc. EvoCOP 2008,
  Naples, Italy}, volume 4972 of \emph{Lecture Notes in Computer Science},
  pages 110--122. Springer Berlin Heidelberg, 2008.
\newblock URL \url{http://dx.doi.org/10.1007/978-3-540-78604-7_10}.

\bibitem[Meisel and Bierwirth(2009)]{meisel2009}
F.~Meisel and C.~Bierwirth.
\newblock Heuristics for the integration of crane productivity in the berth
  allocation problem.
\newblock \emph{Transportation Research Part E: Logistics and Transportation
  Review}, 45\penalty0 (1):\penalty0 196 -- 209, 2009.
\newblock ISSN 1366-5545.
\newblock URL \url{http://dx.doi.org/10.1016/j.tre.2008.03.001}.

\bibitem[Nemhauser and Weber(1979)]{nemhauser1979optimal}
G.~L. Nemhauser and G.~M. Weber.
\newblock Optimal set partitioning, matchings and lagrangian duality.
\newblock \emph{Naval Research Logistics (NRL)}, 26\penalty0 (4):\penalty0
  553--563, 1979.
\newblock URL \url{http://dx.doi.org/10.1002/nav.3800260401}.

\bibitem[Pentico(2007)]{pentico2007assignment}
D.~W. Pentico.
\newblock Assignment problems: A golden anniversary survey.
\newblock \emph{European Journal of Operational Research}, 176\penalty0
  (2):\penalty0 774--793, 2007.
\newblock URL \url{http://dx.doi.org/10.1016/j.ejor.2005.09.014}.

\bibitem[Rezanova and Ryan(2010)]{Rezanova2010}
N.~J. Rezanova and D.~M. Ryan.
\newblock The train driver recovery problem—a set partitioning based model
  and solution method.
\newblock \emph{Computers \& Operations Research}, 37\penalty0 (5):\penalty0
  845 -- 856, 2010.
\newblock ISSN 0305-0548.
\newblock URL \url{http://dx.doi.org/10.1016/j.cor.2009.03.023}.

\bibitem[Ribeiro and Maniezzo(2015)]{itor2015issue}
C.~C. Ribeiro and V.~Maniezzo.
\newblock Preface to the special issue on matheuristics: Model-based
  metaheuristics.
\newblock \emph{International Transactions in Operational Research},
  22\penalty0 (1):\penalty0 1--1, 2015.
\newblock ISSN 1475-3995.
\newblock URL \url{http://dx.doi.org/10.1111/itor.12142}.

\bibitem[Ryan(1992)]{ryan1992}
D.~M. Ryan.
\newblock The solution of massive generalized set partitioning problems in
  aircrew rostering.
\newblock \emph{Journal of the operational research society}, 43\penalty0
  (5):\penalty0 459--467, 1992.
\newblock URL \url{https://doi.org/10.2307/2583565}.

\bibitem[Seb{\H{o}}(1998)]{SeboIPCO1998}
A.~Seb{\H{o}}.
\newblock Characterizing noninteger polyhedra with 0--1 constraints.
\newblock In \emph{Integer Programming and Combinatorial Optimization: 6th
  International IPCO Conference}, volume 1412 of \emph{Lecture Notes in
  Computer Science}, pages 37--52. Springer Berlin Heidelberg, 1998.
\newblock ISBN 978-3-540-69346-8.
\newblock URL \url{http://dx.doi.org/10.1007/3-540-69346-7_4}.

\bibitem[Shor et~al.(1996)Shor, Voitishin, and Glushkov]{shor1996using}
N.~Z. Shor, Y.~V. Voitishin, and V.~Glushkov.
\newblock Using dual network bounds in algorithms for solving generalized set
  packing partitioning problems.
\newblock \emph{Computational Optimization and Applications}, 6\penalty0
  (3):\penalty0 293--303, 1996.
\newblock URL \url{https://doi.org/10.1007/BF00247796}.

\bibitem[Stahlbock and Vo\ss(2008)]{Stahlbock2008}
R.~Stahlbock and S.~Vo\ss.
\newblock Operations research at container terminals: A literature update.
\newblock \emph{OR Spectrum}, 30\penalty0 (1):\penalty0 1--52, 2008.
\newblock URL \url{http://link.springer.com/article/10.1007/s00291-007-0100-9}.

\bibitem[Stefanello et~al.(2015)Stefanello, de~Araújo, and
  Müller]{stefanello2015}
F.~Stefanello, O.~C.~B. de~Araújo, and F.~M. Müller.
\newblock Matheuristics for the capacitated p-median problem.
\newblock \emph{International Transactions in Operational Research},
  22\penalty0 (1):\penalty0 149--167, 2015.
\newblock ISSN 1475-3995.
\newblock URL \url{http://dx.doi.org/10.1111/itor.12103}.

\end{thebibliography}

\end{document}